\documentclass[notitlepage,superscriptaddress,nofootinbib,aps,showpacs,preprint,prd]{revtex4}
\usepackage[latin1]{inputenc}
\usepackage{graphicx}
\usepackage{amssymb}
\usepackage{color}
\usepackage{float}
\usepackage{amsmath}
\usepackage{amsfonts}
\usepackage{dcolumn}
\usepackage{hyperref}
\usepackage{amsthm}
\usepackage{color}

\def\3nab{\tilde{\nabla}}

\def\be {\begin{equation}}
\def\ee {\end{equation}}
\def\ba {\begin{eqnarray}}
\def\ea {\end{eqnarray}}

\newtheorem{prop}{Proposition}

\newcommand{\barray}{\begin{array}}
\newcommand{\earray}{\end{array}}

\newcommand{\bse}{\begin{subequations}} \newcommand{\ese}{\end{subequations}}
 
\begin{document}
\title{Matching conditions in Locally Rotationally Symmetric spacetimes and radiating stars}
\author{Pretty N. Khambule}
\email{khambuleprettyn@yahoo.com}
\affiliation{Astrophysics Research Centre, School of Mathematics, Statistics and Computer Science, University of KwaZulu-Natal, Private Bag X54001, Durban 4000, South Africa.}
 \author{Rituparno Goswami}
\email{Goswami@ukzn.ac.za}
\affiliation{Astrophysics Research Centre, School of Mathematics, Statistics and Computer Science, University of KwaZulu-Natal, Private Bag X54001, Durban 4000, South Africa.}
\author{Sunil D. Maharaj}
\email{Maharaj@ukzn.ac.za}
\affiliation{Astrophysics Research Centre, School of Mathematics, Statistics and Computer Science, University of KwaZulu-Natal, Private Bag X54001, Durban 4000, South Africa.}

\begin{abstract}
We recast the well known Israel-Darmois matching conditions for Locally Rotationally Symmetric (LRS-II) spacetimes using the semitetrad 1+1+2 covariant formalism. This demonstrates how the geometrical quantities including the volume expansion, spacetime shear, acceleration and Weyl curvature of two different spacetimes are related at a general matching surface inheriting the symmetry, which can be timelike or spacelike. 
The approach is purely geometrical and depends on matching the Gaussian curvature of 2-dimensional sheets at the matching hypersurface. This also provides the constraints on the thermodynamic quantities on each spacetime so that they can be matched smoothly across the surface. As an example we regain the Santos boundary conditions and model of a radiating star matched to a Vaidya exterior in general relativity. 
\end{abstract}
 
\maketitle

\section{Introduction\label{Intro}}
Matching two regions of spacetime across a hypersurface has a long history in general relativistic setting, with several researchers contributing towards the understanding of the nature of matching conditions. One of the principal astrophysical application of this study is modelling a radiating star that must satisfy the Einstein field equations across all regions, without having a discontinuity at the boundary that may make the system unstable. For a review of known results with variety of physical applications see \cite{Bonnor,Herr1,Herr2,Santos}.  An important advance was made by Santos \cite{Santos} who showed that an internal spherically symmetric heat conducting source can be matched across a comoving timelike hypersurface to an external radiating atmosphere modelled by the Vaidya geometry. The internal matter distribution can be generalised to include anisotropic stresses \cite{Herr3,Chan,Herr4,Gov1} and also an electromagnetic field \cite{AKG,Ban,Tik,Gov2}. Similarly, the external radiating atmosphere can also be generalised to a specific combination of radiation and perfect fluids, modelled by the generalised Vaidya metric \cite{Gov3}. Several exact models of radiating stars have been found in recent times \cite{Iva1,Iva2,Iva3,Abebe,Maharaj1,Maharaj2}. Also Abebe and Maharaj \cite{Abebe}, amongst others, have applied the Lie analysis of the differential equations to find new dissipating models with an equation of state. All of these radiating models depend critically on the Santos boundary condition, which is a non-linear differential equation. \\

The matching of two spacetime regions requires that the field equations are continuous across the boundary. Just like in electromagnetism we match the potential and the normal derivative of the potential across the boundary so that the Maxwell equations remain continuous. In general relativity, we match the gravitational potential (the metric) and the Lie derivative along the normal (the extrinsic curvature). In general this is a complicated process that requires introducing intrinsic coordiantes on the hypersurface, calculation of extrinsic curvature and application of the Einstein field equations. The analysis is extremely technical and it is not surprising that researchers largely restrict attention to the spherically symmetric spacetimes with specific matter distributions. The main purpose of this paper are as follows: Firstly we show that a proper matching is possible for a wide class of spacetimes, of which spherical symmetry is a special case. This is achieved in Locally Rotationally symmetric-II (LRS-II)  class of spacetimes with general matter distributions. Secondly we seek to simplify the matching process and use only the Israel-Darmois matching conditions via local semitetrad decomposition of the spacetimes. Thirdly, we establish a purely geometrical basis for the matching. We find the Gaussian curvature of the two dimensional foliation of the matching hypersurface is one of the fundamental quantities that determines the matching. The physically important result of Santos \cite{Santos} is then recovered as a special case. \\

The paper is organised as follows: In section 2 we discuss the Israel-Darmois matching conditions. In the next section we introduce the geometrical semitetrad 1+1+2 decomposition of the spacetimes and write down the field equations for LRS-II spacetimes. In section 4 we write down the matching conditions in terms of the geometrical and thermodynamic 1+1+2 quantities, and their consistency for timelike matching hypersurface. In the next section we follow the same procedure for a spacelike matching hypersurface for completeness. And finally, in section 6, we use this geometrical approach to recover the well known result of matching a radiating spherical star with a Vaidya exterior.

\section{Israel-Darmois matching conditions \label{one}}

Historically, the matching conditions of two spacetimes across a hypersurface was first worked out by Darmois \cite{Darmois}, which was later further developed by Lichnerowicz \cite{Lich}, Israel \cite{Israel}, Clarke and Dray \cite{Clarke}, and Mars, Senovilla and Fayos \cite{Seno,Fayos}. The key idea of this matching problem is as follows: Let us have two disjoint spacetime manifolds,  $\mathcal{V}^+(g^+)$ with an imbedded boundary 3-hypersurface $\mathcal{S}^+$ and $\mathcal{V}^-(g^-)$ with an imbedded boundary 3-hypersurface $\mathcal{S}^-$. We further assume that there is a $\mathcal{C}^3$ diffeomorphism from $\mathcal{S}^-$ to $\mathcal{S}^+$. This means that there is a three times continuously differentiable invertible function which maps from $S^-$ to $S^+$. Let the disjoint union of $\mathcal{V}^-$ to $\mathcal{V}^+$, which have points that are related through the diffeomorphism identified, be the complete spacetime, which we shall denote as $\mathcal{V}_4(g)$. The images  of $\mathcal{S}^-$ and $\mathcal{S}^+$ in $\mathcal{V}_4$ shall be noted by $\mathcal{S}$. This issue now is if  $\mathcal{V}^+$ and $\mathcal{V}^-$ can be joined in such a manner that $\mathcal{V}_4$ has a  Lorentzian geometry with Einstein field equations well defined. As shown clearly in \cite{Israel,Clarke}, this is possible if and only if  $\mathcal{S}^+$ and $\mathcal{S}^-$ are isometrical with respect to their first fundamental form $h^+$ and $h^-$ (the induced metric on the hypersurface via the imbedding) which have been derived from $g^+$ and $g^-$ respectively, as in this case there is a natural continuous extension of the metric $g$ to the entire $\mathcal{V}_4$.\\

Thus, from the point of view of $\mathcal{V}^+(g^+)$ and $\mathcal{V}^-(g^-)$, there are two imbeddings: $x^{\mu}_{\pm}=x^{\mu}_{\pm}(\xi^a)$ of $\mathcal{S}$, where $\xi^a$ are intrinsic coordinates for  $\mathcal{S}$ and $x^{\mu}_{\pm}$ are local coordinates for $\mathcal{V}^{\pm}$. The requirement that the first fundamental forms must match is
\begin{eqnarray}\label{FFF}
h^{+}_{ab}=h^{-}_{ab}\;,
\end{eqnarray}
where from \cite{Israel,Clarke,Seno}
\begin{eqnarray}
h^{\pm}_{ab}=g^{\pm}_{\mu v}(x_{\pm}(\xi))\frac{\partial x^{\mu}_{\pm}(\xi) }{\partial \xi^a}\frac{\partial x^{\mu}_{\pm}(\xi)}{\partial \xi^b}.
\end{eqnarray}
Note that $h_{ab}$ is the 3-metric on $\mathcal{S}$.\\

Equation (\ref{FFF}) is an important condition for the calculation of the Riemann tensor distribution and its contractions. The singular part of this tensor distribution is proportional to the Dirac one-form distribution $\delta_\mu$ which is linked with $\mathcal{S}$. Therefore this singular part describes an {\it infinite discontinuity} at $\mathcal{S}$. For a smooth matching, these infinite discontinuities needs to be avoided in  matter and curvature tensors. This occurs if and only if the second fundamental form of $\mathcal{S}$ match, that is
\begin{eqnarray}\label{SFF1}
\chi^{-}_{ab}=\chi^{+}_{ab},
\end{eqnarray}
where
\begin{eqnarray}\label{SFF2}
\chi^{\pm}_{ab}=-n^{\pm}_{\mu}\left(\frac{\partial^2x^{\mu}_{\pm}(\xi)}{\partial\xi^a\partial\xi^b}+\Gamma^{\pm \mu}_{pv}\frac{\partial x^{p}_{\pm}(\xi)}{\partial \xi^a}\frac{\partial x^v_{\pm}(\xi)}{\partial\xi^b}\right).
\end{eqnarray}
Note that $\Gamma^{\mu}_{\rho \nu}$ is  the Christoffel symbol of the second kind, and it represents the metric connection coefficients which are given by
\begin{eqnarray}
\Gamma^{\mu}_{\rho \nu}=\frac{1}{2}g^{\mu\lambda}(g_{\nu\lambda,\rho}+g_{\lambda \rho,\nu}-g_{\rho \nu,\lambda}),
\end{eqnarray}
where a comma denotes partial differentiation. Thus to match two spacetimes across their common boundary, the matching conditions (\ref{FFF}) and (\ref{SFF1}) must be satisfied.
\\

In the subsequent sections we recast the above matching conditions in terms of geometrical and thermodynamic variables for LRS-II spacetimes, using semi-tetrad 1+1+2 covariant formalism. This gives a beautiful physical interpretation of the quantities that should be continuous across the matching surface $\mathcal{S}$, where $\mathcal{S}$ can be timelike or spacelike.

\section{Semi-tetrad decomposition of LRS-II spacetimes}\label{two}

A spacetime manifold ($\mathcal{M},g$) is called $\textit{locally isotropic}$, if every point  $p \in (\mathcal{M},g)$ has a continuous non-trivial isotropy group. When this group consists of spatial rotations the spacetime is called \textit{locally rotationally symmetry } or LRS \cite{ElstEllis}. Within LRS spacetime, there exists a unique, preferred spatial direction at each point and this preferred direction is covariantly  defined. This direction results in a local axis of symmetry, such that all  observations are identical under rotation about it. LRS-II is a subclass of LRS spacetimes that is free of rotation. By the symmetry of LRS-II, we can covariantly decompose the spacetimes using using a unit timelike vector $u^a$ along the fluid flow lines and a unit spacelike vector $e^a$ along the preferred spatial direction.
\\

With respect to a timelike congruence, the spacetime can be locally decomposed into space and time. The timelike congruence is defined by the flow lines, with the \textit{four-velocity} as
\begin{eqnarray}\label{0.1}
u^a=\tfrac{dx^a}{d\tau},&& u^au_a=-1,
\end{eqnarray}
where $\tau$ is the \textit{proper time} measured along the flow lines. Given  the 4-velocity $u^a$, we have  unique parallel and orthogonal \textit{projection tensor}
\begin{equation}
U^a_b=- u^au_b,\;\;{\rm{h}}^a_b=g^a_b+u^au_b,
\end{equation}
where ${\rm h}^a_b$ is the projection tensor that projects any 4-d vector or tensor onto the 3-space orthogonal to `$u^a$' . The volume element of this 3-space is given as $\epsilon_{abc}=\eta_{abcd}u^d$, where $\eta_{abcd}$ is the usual volume element of 4-space.

This decomposition naturally gives two directional derivatives: the vector $u^a$ is used  to define the  covariant time derivative along the flow lines (denoted by dot) for any tensor $S^{a..b}{}_{c..d}$, given by
\begin{eqnarray}
\dot{S}^{a..b}{}_{c..d}=u^e\nabla_eS^{a..b}{}_{c..d},
\end{eqnarray}
and the tensor ${\rm h}_{ab}$ is used to define the fully \textit{orthogonally projected covariant} derivative D for any tensor $S^{a..b}{}_{c..d}$
\begin{eqnarray}
D_eS^{a..b}{}_{c..d}= {\rm h}^a_f {\rm h}^p_c...{\rm h}^b_g{\rm h}^q_d{\rm h}^r_e\nabla_rS^{f..g}_{p..q},
\end{eqnarray}
with total projection on all free indices. Therefore, the covariant derivative of $u^a$ is decomposed as
\begin{eqnarray}
\nabla_au_b=-u_aA_b+\tfrac{1}{3}\Theta {\rm h}_{ab}+\sigma_{ab}+\epsilon_{abc}\omega^c,
\end{eqnarray}
where $A_{b}=\dot{u}_b$ is the acceleration, $\Theta=D_au^a$ is the expansion, $\sigma_{ab}=D_{\langle a}u_{b\rangle}$ is the shear tensor that denotes the distortion and $\omega^c$ is the vorticity vector denoting the rotation. The Weyl curvature tensor is also split relative to $u^a$ into the \textit{electric} and \textit{magnetic Weyl curvature} parts as 
\begin{eqnarray}
E_{ab}=C_{abcd}u^cu^d=E_{\langle ab \rangle};&& H_{ab}=\tfrac{1}{2}\epsilon_{ade}C^{de}_{bc}u^c=H_{\langle ab \rangle}.
\end{eqnarray}
Here the angle brackets denote the projected trace-free part. Similarly, the energy momentum tensor of matter is decomposed as follows:
\begin{eqnarray}
T_{ab}=\mu u_au_b+q_au_b+q_bu_a+p{\rm h}_{ab}+\pi_{ab},
\end{eqnarray}
where $\mu=T_{ab}u^au^b$ is the energy density, $p=(1/3){\rm h}^{ab}T_{ab}$ is the isotropic pressure, $q_a=q_{\langle a\rangle}=-{\rm h}^c_aT_{cd}u^d$ is the 3-vector that defines the heat flux and $\pi_{ab}=\pi_{\langle ab \rangle}$ is the anisotropic stress.\\

The 1+1+2 decomposition of spacetime is a natural extension of the 1+3 formalism in which the 3-space is further decomposed with respect to a given spatial direction. In the 1+1+2 approach, the spacetime is further split through the use of a preferred spatial vector  $e^a$  which is orthogonal to $u^a$. We now choose a spacelike  vector field $e^a$ such that 
\begin{eqnarray}\label{0.2}
e^ae_a=1, && u^ae_a=0\;.
\end{eqnarray}
 The new projection tensor is given by 
\begin{eqnarray}
N^a_b\equiv {\rm h}^a_b-e^ae_b=g^a_b+u^au_b-e^ae_b
\end{eqnarray}
which project vectors orthogonal to $e^a$ and $u^a$ onto a 2-surface called \textit{sheet}. Thus
\begin{eqnarray}
 e^aN_{ab}=0=u^aN_{ab}, && N^a_a=2.
\end{eqnarray} 
This 1+1+2 splitting of the spacetime gives rise to the new directional derivatives along $e^a$ and on the 2-surface:
\begin{itemize}
	\item The\textit{ hat derivative} is the spatial derivative along the vector field $e^a$: for any 3-tensor $\psi_{a..b}{}^{c..d}$, \ $\hat{\psi}_{a..b}{}^{c..d}\equiv e^fD_f\psi_{a..b}{}^{c..d}$.
	\item The \textit{delta derivative} is the projected derivative onto the sheet by $N_a^b$, with  projection on all free indices: for any 3-tensor $\psi_{a..b}{}^{c..d}$, \  $\delta_e\psi_{a..b}
	{}^{c..d}\equiv N_a^f..N_b^gN_h^c..N_i^dN_e^jD_j\psi_{f..g}{}^{h..i}$. 
\end{itemize}
In the 1+1+2 splitting , the  4-acceleration, vorticity and shear split as follows:
\begin{eqnarray}
\dot{u}^a&=&\mathcal{A}e^a+\mathcal{A}^a,\label{0.3}\\
\omega^a&=&\Omega e^a+\Omega^a,\\
\sigma_{ab}&=&\Sigma(e_ae_b-\tfrac{1}{2}N_{ab})+2\Sigma_{(a} e_{b)}+\Sigma_{ab}.
\end{eqnarray}
For the electric and magnetic  Weyl tensors we get 
\begin{eqnarray}
E_{ab}&=&\mathcal{E}(e_ae_b-\tfrac{1}{2}N_{ab})+2\mathcal{E}_{(a}e_{b)}+\mathcal{E}_{ab},\\
H_{ab}&=&\mathcal{H}(e_ae_b-\tfrac{1}{2}N_{ab})+2\mathcal{H}_{(a}e_{b)}+\mathcal{H}_{ab}.
\end{eqnarray}
Similarly, the fluid variable $q^a$ and $\pi_{ab}$ may be split as
\begin{eqnarray}
q^a&=&Qe^a+Q^a,\\
\pi_{ab}&=&\Pi(e_ae_b-\tfrac{1}{2}N_{ab})+\Pi_{(a}e_{b)}+\Pi_{ab}.
\end{eqnarray}

We can decompose the covariant derivative of $e^a$ in the direction orthogonal to $u^a$ into it's irreducible parts giving 
\be 
{\rm D}_{a}e_{b} = e_{a}a_{b} + \frac{1}{2}\phi N_{ab} + 
\xi\epsilon_{ab} + \zeta_{ab}~.
\ee
We see that on the 3-space,  moving along the preferred vector $ e^{a} $, $\phi$ represents the \textit{expansion of the sheet},  $\zeta_{ab}$ is the \textit{shear of $e^{a}$} (i.e. the distortion of the sheet) and $a^{a}$ its \textit{acceleration}. We can also interpret $\xi$ as the \textit{vorticity} associated with $e^{a}$ so that it is a representation of the ``twisting'' or rotation of the sheet.

The full covariant derivatives of vector fields $e_a$ and $u_a$ can now split into 
\begin{eqnarray}
\nabla_{a}u_b&=&-u_a(\mathcal{A}e_b+\mathcal{A}_b)+e_ae_b(\tfrac{1}{3}\Theta+\Sigma)+\Omega\varepsilon_{ab}\nonumber\\
 &&+ e_a(\Sigma_{b}+\varepsilon_{bc}\Omega^c)+e_b(\Sigma_a-\epsilon_{ac}\Omega^c)\nonumber\\
&&+N_{ab}(\tfrac{1}{3}\Theta-\tfrac{1}{2}\Sigma)+\Sigma_{ab},\label{1.1}\\
\nabla_{a}e_b&=&-\mathcal{A}u_au_b-u_a\alpha_b+e_au_b(\tfrac{1}{3}\Theta+\Sigma)+\xi\varepsilon_{ab}\nonumber\\
&&+u_b(\Sigma_a-\epsilon_{ac}\Omega^c)+e_aa_b+\tfrac{1}{2}\phi N_{ab}+\zeta_{ab}.\label{1.2}
\end{eqnarray}
By symmetry of LRS-II spacetimes we can easily see that the sheet components of all the vectors and tensor quantities vanish identically. Thus, the set quantities that fully describe LRS-II spacetime are $\left\lbrace\mathcal{A},\Theta,\phi,\Sigma,\mathcal{E},\mu,p,\Pi,Q\right\rbrace$.\\

Decomposing the Ricci identities for $u^a$ and $e^a$ and the doubly contracted Bianchi identities, we can get the following field equation for LRS spacetime:\\\\
\textit{Propagation}:
\begin{eqnarray}
\hat{\phi}&=&-\tfrac{1}{2}\phi^2+(\tfrac{1}{3}\Theta+\Sigma)(\tfrac{2}{3}\Theta-\Sigma)\nonumber\\
&&-\tfrac{2}{3}(\mu+\Lambda)-\mathcal{E}-\tfrac{1}{2}\Pi,\\
\hat{\Sigma}-\tfrac{2}{3}\hat{\Theta}&=&-\tfrac{3}{2}\phi\Sigma-Q,\\
\hat{\mathcal{E}}-\tfrac{1}{3}\hat{\mu}+\tfrac{1}{2}\hat{\Pi}&=&-\tfrac{3}{2}\phi(\mathcal{E}+\tfrac{1}{2}\Pi)+(\tfrac{1}{2}\Sigma-\tfrac{1}{3}\Theta)Q.
\end{eqnarray}
\textit{Evolution}:
\begin{eqnarray}
\dot{\phi}&=&-(\Sigma-\tfrac{2}{3}\Theta)(\mathcal{A}-\tfrac{1}{2}\phi)+Q,\\
\dot{\Sigma}-\tfrac{2}{3}\dot{\Theta}&=&-\mathcal{A}\phi+2(\tfrac{1}{3}\Theta-\tfrac{1}{2}\Sigma)^2\nonumber\\
&&+\tfrac{1}{3}(\mu+3p-2\Lambda)-\mathcal{E}+\tfrac{1}{2}\Pi,\\
\dot{\mathcal{E}}-\tfrac{1}{3}\dot{\mu}+\tfrac{1}{2}\dot{\Pi}&=&+(\tfrac{3}{2}\Sigma-\Theta)\mathcal{E}+\tfrac{1}{4}(\Sigma-\tfrac{2}{3}\Theta)\Pi\nonumber\\
&&+\tfrac{1}{2}\phi Q-\tfrac{1}{2}(\mu+p)(\Sigma-\tfrac{2}{3}\Theta).
\end{eqnarray} 
\textit{Propagation/evolution}:
\begin{eqnarray}
\mathcal{\hat{A}}-\dot{\Theta}&=&- (\mathcal{A}+\phi)\mathcal{A}+\tfrac{1}{3}\Theta^2+\tfrac{3}{2}\Sigma^2\nonumber\\
&&+\tfrac{1}{2}(\mu+3p-2\Lambda),\\
\dot{\mu}+\hat{Q}&=&-\Theta(\mu+p)-(\phi+2\mathcal{A})Q-\tfrac{3}{2}\Sigma\Pi\\
\dot{Q}+\hat{p}+\hat{\Pi}&=&-(\tfrac{3}{2}\phi+\mathcal{A})\Pi-(\tfrac{4}{3}\Theta+\Sigma)Q\nonumber\\
&&-(\mu+p)\mathcal{A}.
\end{eqnarray}
The Gaussian curvature $K$ of the 2-sheet is defined in terms of the Ricci scalar of the two sheet as $^2R_{ab}=KN_{ab}$ and can be written in terms of the covariant scalars as
\begin{eqnarray}
K=\tfrac{1}{3}(\mu+\Lambda)-\mathcal{E}-\tfrac{1}{2}\Pi+\tfrac{1}{4}\phi^2-(\tfrac{1}{3}\Theta-\tfrac{1}{2}\Sigma)^2.
\end{eqnarray}
 Thus, the evolution and propagation equation for the Gaussian curvature $K$ are
 \begin{eqnarray}
\dot{ K}&=&-(\tfrac{2}{3}\Theta-\Sigma)K,\\
\hat{ K}&=&-\phi K.
 \end{eqnarray}
 
\section{Matching Conditions for LRS-II spacetimes: Timelike matching surface}\label{four}

Let us consider two regions of spacetime, both having LRS-II symmetry and they are matched across a non-compact hypersurface $S$. We consider $S$ to inherit the LRS II  symmetry, in the sense that all sheet components vanish identically on $S$. In the case of spherical symmetry we can then easily describe $S$ as a curve on $[u,e]$ plane, with each point on the curve representing the surface of a 2-sheet. We describe the two matching regions as region 1 and region 2 respectively. All the geometrical and thermodynamic quantities in region 2 will be denoted by the usual variables with a tilde. For example:
the metric tensor in region 1 is given by
\begin{align}\label{1}
g_{ab}=-u_{a}u_{b}+e_{a}e_{b}+N_{ab}.
\end{align}
while that of region 2 is given by
\begin{align}\label{2}
\tilde{g}_{ab}=-\tilde{u}_{a}\tilde{u}_{b}+\tilde{e}_{a}\tilde{e}_{b}+\tilde{N}_{ab}.
\end{align}
We will now consider two distinct cases: (a) when the normal to the matching hypersurface $S$ is spacelike and (b) when it is timelike. We will explicitly write down the Israel-Darmois condition for all two cases and extract the scalar equations in terms of the geometrical variablles of LRS-II spacetimes from them. This will then transparently show us, the behaviour of these geometrical and thermodynamic quantities across the hypersurface.  

\begin{table}[!h]
  \caption{Geometry of Region 1 and Region 2}
\begin{center}
\begin{tabular}{|c|c|}\hline\hline
 {\bf REGION 1}& {\bf REGION 2} \\ \hline
$g_{ab}=-u_{a}u_{b}+e_{a}e_{b}+N_{ab}$ &  $\tilde{g}_{ab}=-\tilde{u}_{a}\tilde{u}_{b}+\tilde{e}_{a}\tilde{e}_{b}+\tilde{N}_{ab}$\\ \hline
$n_{a}=\alpha u_{a}+\beta e_{a}$ ; $\beta=\pm\sqrt{1+\alpha^2}$ & $\tilde{n}_{a}=\tilde{\alpha} \tilde{u}_{a}+\tilde{\beta} \tilde{e}_{a}$ ; $\tilde\beta=\pm\sqrt{1+\tilde\alpha^2}$\\ \hline
$h_{ab}= g_{ab}-n_{a}n_{b}$ & $\tilde{h}_{ab}=\tilde{g}_{ab}-\tilde{n}_{a}\tilde{n}_{b}$\\ \hline
$\chi_{ab}= h^{c}_{(a}h^{d} _{b)}\nabla_{d}n_{c} $ & $\tilde{\chi}_{ab}= \tilde{h}^{c}  _{(a}\tilde{h}^{d} _{b)}\nabla_{d}\tilde{n}_{c}$ \\ \hline
$\dot\Psi=u^a\nabla_a\Psi\;;\;\hat\Psi=e^aD_a\Psi$ & $\mathring{\tilde{\Psi}}=\tilde{u}^a\nabla_{a}\tilde{\Psi}\;;\; \bar{\tilde{\Psi}}=\tilde{e}^aD_a\tilde{\Psi}$\\ \hline
\hline
\end{tabular}
\end{center}
\end{table}

Let us first suppose the unit normal to $S$ is spacelike. This situation corresponds to the boundary matching hypersurface of a dynamic spherical star with the exterior spherically symmetric spacetime.
Let $n_{a}$ be the unit normal in Region 1 to the matching timelike hypersurface $S$. Since $S$ is taken to inherit the LRS-II  symmetry, the unit normal is given as
\begin{align}\label{3}
n_{a}=\alpha u_{a}+\beta e_{a},
\end{align}
where $\alpha$ and $\beta$ are functions of the curve parameters of integral curves of $u^a$ and $e^a$ respectively, with the normalisation condition $\beta=\pm\sqrt{1+\alpha^2}$.
Similarly in  Region 2, let $\tilde{n}_{a}$ be the unit normal to $S$. Thus we have 
\begin{align}\label{4}
\tilde{n}_{a}=\tilde{\alpha} \tilde{u}_{a}+\tilde{\beta} \tilde{e}_{a},
\end{align}
where $\tilde\alpha$ and $\tilde\beta$ are functions of the curve parameters of integral curves of $\tilde u^a$ and $\tilde e^a$ respectively, with the normalisation condition $\tilde\beta=\pm\sqrt{1+\tilde\alpha^2}$. Since both $n_{a}$ and $\tilde{n}_{a}$ are spacelike, we make use of \eqref{1} and \eqref{3}, to get the first fundamental form on the boundary $S$ is in Region 1 as
\begin{align}
h_{ab}&= g_{ab}-n_{a}n_{b}\nonumber\\
&=-(1+\alpha^{2})u_{a}u_{b}+(1-\beta^{2})e_{a}e_{b}-\alpha\beta u_{a}e_{b}-\alpha\beta e_{a}u_{b}+N_{ab},\label{6}
\end{align}
where $\beta=\pm\sqrt{1+\alpha^2}$.
Likewise, using \eqref{2} and \eqref{4}, the first fundamental form on $S$ in Region 2 is	
\begin{align}
\tilde{h}_{ab}&=\tilde{g}_{ab}-\tilde{n}_{a}\tilde{n}_{b}\label{7}\\
&=-(1+\tilde{\alpha}^{2})\tilde{u}_{a}\tilde{u}_{b}+(1-\tilde{\beta}^{2})\tilde{e}_{a}\tilde{e}_{b}-\tilde{\alpha}\tilde{\beta }\tilde{u}_{a}\tilde{e}_{b}-\tilde{\alpha}\tilde{\beta } \tilde{e}_{a}\tilde{u}_{b}+\tilde{N}_{ab},\label{8}
\end{align}
with the constraint $\tilde\beta=\pm\sqrt{1+\tilde\alpha^2}$. 
	
Now the second fundamental form in Regions 1 and 2 are defined as follows:
\begin{align}\label{10}
\chi_{ab}= h^{c}_{(a}h^{d} _{b)}\nabla_{d}n_{c} \;\;;\;\; \tilde{\chi}_{ab}= \tilde{h}^{c}  _{(a}\tilde{h}^{d} _{b)}\nabla_{d}\tilde{n}_{c}.
\end{align}
Noting that for any scalar $\lambda$ in a spacetime with LRS-II symmetry, we can write
\begin{align}\label{17}
\nabla_{a}\lambda=-\dot{\lambda}u_{a}+\hat{\lambda}e_{a},
\end{align}
Using \eqref{17} we can immediately see that
\begin{align}\label{20}
\nabla_{b}n_{a}&=\left(-\dot{\alpha}-\beta\mathcal{A}\right)u_{a}u_{b}+\left[\alpha\left(\tfrac{1}{3}\Theta+\Sigma\right)+\beta\right]e_{a}e_{b}+\left[\beta\left(\tfrac{1}{3}\Theta+\Sigma\right)+\hat{\alpha}\right]u_{a}e_{b}\nonumber\\
&+\left(-\dot{\beta}-\alpha\mathcal{A}\right)e_{a}u_{b}+\left[\alpha\left(\tfrac{1}{3}\Theta-\tfrac{1}{2}\Sigma\right)+\tfrac{1}{2}\beta\phi\right]N_{ab}.
\end{align} 
Therefore using \eqref{6} and \eqref{10}, we have 
\begin{eqnarray}\label{21}
\chi_{ab}&=&u_{a}u_{b}\left[-(1+\alpha^{2})^2\dot{\alpha}+\alpha\beta(1+\alpha^{2})\dot{\beta}-\alpha\beta(1+\alpha^{2})\hat{\alpha}\right.\nonumber\\
&&\left.+\alpha^{2}\beta^{2}\hat{\beta}-\beta \mathcal{A}(1+\alpha^{2})-\alpha\beta^{2}\left(\tfrac{1}{3}\Theta+\Sigma\right)\right]+e_{a}e_{b}\left[-\alpha^{2}\beta^{2}\dot{\alpha}\right.\nonumber\\
&&\left.-\alpha\beta(1-\beta^{2})\dot{\beta}+\alpha\beta(1-\beta^{2})\hat{\alpha}+(1-\beta^{2})^{2}\hat{\beta}-\alpha^{2}\beta\mathcal{A}\right.\nonumber\\
&&\left.+\alpha(1-\beta^{2})(\tfrac{1}{3}\Theta+\Sigma)\right]+ u_{(a}e_{b)}\left[-2\alpha\beta(1+\alpha^{2})\dot{\alpha}-2\alpha\beta(1-\beta^2)\hat{\beta}\right.\nonumber\\
&&\left.+(-\alpha^2\beta^2+(1+\alpha^2)(1-\beta^2))\hat{\alpha}+(\alpha^2\beta^2-(1+\alpha^2)(1-\beta^2))\dot{\beta}\right.\nonumber\\
&&\left.-\alpha\mathcal{A}(\beta^2+(1+\alpha^2))+(\beta(1-\beta^{2})-\alpha^2\beta)\left(\tfrac{1}{3}\Theta+\Sigma\right)\right]\nonumber\\
&&+N_{ab}\left[\alpha\left(\tfrac{1}{3}\Theta-\tfrac{1}{2}\Sigma\right)+\tfrac{1}{2}\beta\phi\right],
\end{eqnarray}
with the constraint $\beta=\pm\sqrt{1+\alpha^2}$. We have a similar result for Region 2. First we replace the dot and hat derivative with the circle and bar derivative respectively. The circle derivative arises from using the operator $\tilde{u}^a\nabla_{a}$ while the bar derivative arises from using the operator $\tilde{e}^aD_a$. Thus for Region 2 \eqref{17} becomes
\begin{align}\label{a1}
\nabla_{a}\tilde{\lambda}=-\mathring{\tilde{\lambda}}u_{a}+\bar{\tilde{\lambda}}e_{a},
\end{align}
Applying \eqref{a1}, \eqref{10} becomes
\begin{eqnarray}\label{22}
\tilde{\chi}_{ab}&=&\tilde{u}_{a}\tilde{u}_{b}\left[-(1+\tilde{\alpha}^{2})^2\mathring{\tilde{\alpha}}+\tilde{\alpha}\tilde{\beta}(1+\tilde{\alpha}^{2})\mathring{\tilde{\beta}}-\tilde{\alpha}\tilde{\beta}(1+\tilde{\alpha}^{2})\bar{\tilde{\alpha}}\right.\nonumber\\
&&\left.+\tilde{\alpha}^{2}\tilde{\beta}^{2}\bar{\tilde{\beta}}-\tilde{\beta} \mathcal{\tilde{A}}(1+\tilde{\alpha}^{2})-\tilde{\alpha}\tilde{\beta}^{2}\left(\tfrac{1}{3}\tilde{\Theta}+\tilde{\Sigma}\right)\right]+\tilde{e}_{a}\tilde{e}_{b}\left[-\tilde{\alpha}^{2}\tilde{\beta}^{2}\mathring{\tilde{\alpha}}\right.\nonumber\\
&&\left.-\tilde{\alpha}\tilde{\beta}(1-\tilde{\beta}^{2})\mathring{\tilde{\beta}}+\tilde{\alpha}\tilde{\beta}(1-\tilde{\beta}^{2})\bar{\tilde{\alpha}}+(1-\tilde{\beta}^{2})^{2}\bar{\tilde{\beta}}-\tilde{\alpha}^{2}\tilde{\beta}\mathcal{\tilde{A}}\right.\nonumber\\
&&\left.+\tilde{\alpha}(1-\tilde{\beta}^{2})(\tfrac{1}{3}\tilde{\Theta}+\tilde{\Sigma})\right]+ \tilde{u}_{(a}\tilde{e}_{b)}\left[-2\tilde{\alpha}\tilde{\beta}(1+\tilde{\alpha}^{2})\mathring{\tilde{\alpha}}-2\tilde{\alpha}\tilde{\beta}(1-\tilde{\beta}^2)\bar{\tilde{\beta}}\right.\nonumber\\
&&\left.+(-\tilde{\alpha}^2\tilde{\beta}^2+(1+\tilde{\alpha}^2)(1-\tilde{\beta}^2))\bar{\tilde{\alpha}}+(\tilde{\alpha}^2\tilde{\beta}^2-(1+\tilde{\alpha}^2)(1-\tilde{\beta}^2))\mathring{\tilde{\beta}}\right.\nonumber\\
&&\left.-\tilde{\alpha}\mathcal{\tilde{A}}(\beta^2+(1+\tilde{\alpha}^2))+(\tilde{\beta}(1-\tilde{\beta}^{2})-\tilde{\alpha}^2\tilde{\beta})\left(\tfrac{1}{3}\tilde{\Theta}+\tilde{\Sigma}\right)\right]\nonumber\\
&&+\tilde{N}_{ab}\left[\tilde{\alpha}\left(\tfrac{1}{3}\tilde{\Theta}-\tfrac{1}{2}\tilde{\Sigma}\right)+\tfrac{1}{2}\tilde{\beta}\tilde{\phi}\right],
\end{eqnarray}
with the constraint $\tilde\beta=\pm\sqrt{1+\tilde\alpha^2}$.\\

On the boundary hypersurface: the Israel-Darmois matching conditions are now given as
\begin{equation}\label{112Match}
h_{ab}=\tilde{h}_{ab}\;\;;\;\; \chi_{ab}=\tilde{\chi}_{ab}.
\end{equation}

\subsection{Extracting the scalar equations: Matching the geometrical quantities}

To extract the relevant 1+1+2 scalar equations from the above conditions (\ref{112Match}), we note the following important observations: keeping with the symmetry of LRS spacetime on both sides of the matching hypersurface, we must have the metric on the 2 dimensional sheets foliating the matching 3-surface exactly same as calculated from the both sides. That is, we must have 
\begin{align}\label{a4}
N_{ab}=\tilde{N}_{ab},
\end{align}
on the boundary. In other words, the spherical 2-surfaces are the same on the boundary if we approach it from either side. This gives two scalar equations that must be satisfied on the boundary:
\begin{equation}\label{M1}
N^{ab}\chi_{ab}=\tilde{N}^{ab}\tilde{\chi}_{ab},
\end{equation}
and
\begin{equation}
(h^{ab}-N^{ab})\chi_{ab}=(\tilde{h}^{ab}-\tilde{N}^{ab})\tilde{\chi}_{ab}.
\end{equation}

This brings us to the following important proposition:
 
\begin{prop}
The expansion, shear and sheet-expansion in both regions must satisfy the following constraint on the timelike boundary hypersurface $S$,
\begin{eqnarray} 
\beta\phi-\alpha\left(\Sigma-\tfrac{2}{3}\Theta\right)&=\tilde{\beta}\tilde{\phi}-\tilde{\alpha}\left(\tilde{\Sigma}-\tfrac{2}{3}\tilde{\Theta}\right).\label{23}
\end{eqnarray}
\end{prop}
\begin{proof}
The proof follows in a straightforward fashion by plugging in  \eqref{21} and  \eqref{22} in \eqref{M1}
\end{proof}

\subsection{Special case of spherical symmetry}
 
Let us now, consider the special case of spherical symmetry, where the 2-sheets are 2-spheres. In this case, matching the 2-dimensional metric $N_{ab}$ on the natural 2-foliations of the matching hypersurface, naturally matches the intrinsic curvature of the surface when approached from 
either side. As we know, the trace of 2-dimensional Ricci tensor on these foliations is the Gaussian curvature, therefore the Gaussian curvature of the 2-sheets must be same when calculated from either side. In other words, we must have $K=\tilde K$:
 \begin{eqnarray}
\tfrac{1}{3}\mu-\mathcal{E}-\tfrac{1}{2}\Pi+\tfrac{1}{4}\phi^2-(\tfrac{1}{3}\Theta-\tfrac{1}{2}\Sigma)^2=\tfrac{1}{3}\tilde\mu-\mathcal{\tilde E}-\tfrac{1}{2}\tilde\Pi+\tfrac{1}{4}\tilde\phi^2-(\tfrac{1}{3}\tilde\Theta-\tfrac{1}{2}\tilde\Sigma)^2\;.
\end{eqnarray}
Moreover, any scalar constructed from the Gaussian curvature and it's covariant derivative should also be continuous across the matching hypersurface. One such scalar that has a well defined physical significance is the Misner-Sharp mass (the mass within a given 2-sheet at any instant of time), which for any spherically symmetric spacetime is given as
\begin{equation}
\mathcal{M}=\frac{1}{2\sqrt{K}}\left(1-\frac{1}{4K^3}\nabla_aK\nabla^aK\right).
\end{equation}
Using the definition of Gaussian curvature and the field equations, we can immediately get
\begin{equation}
\mathcal{M}=\frac{1}{2(\sqrt{K})^3}\left(\frac{1}{3}\mu-\mathcal{E}-\frac{1}{2}\Pi\right),
\end{equation}
Since $K=\tilde K$ and $\mathcal{M}=\tilde{\mathcal{M}}$, this completes the demonstration of the following propesition:

\begin{prop}
For sperically symmetric spacetimes, the expansion, shear and sheet-expansion in the both regions must satisfy the following constraint on the timelike boundary hypersurface $S$,
\begin{equation}
\tfrac{1}{4}\phi^2-(\tfrac{1}{3}\Theta-\tfrac{1}{2}\Sigma)^2=\tfrac{1}{4}\tilde\phi^2-(\tfrac{1}{3}\tilde\Theta-\tfrac{1}{2}\tilde\Sigma)^2\;.
\label{match2}
\end{equation}
\end{prop}
The above equation, together with (\ref{23}), completely specifies how the volume expansion, shear and 2-sheet expansion are related at either side of the matching hypersurface. 

\subsection{Consistent propagation of the matching constraint}

It is interesting to note the condition (\ref{23}) is a constraint on the matching hypersurface that must be satisfied at all epochs. Thus acting on (\ref{23}) with the operator $n^a\nabla_a=\tilde n^a\nabla_a$ should be identically zero. 
In other words, we must have
\begin{align}\label{a}
(\alpha u^a+\beta e^a)\nabla_{a}\left(\beta\phi-\alpha\left(\Sigma-\tfrac{2}{3}\Theta\right)\right)=(\tilde{\alpha}\tilde{u}^a+\tilde{\beta}\tilde{e}^a)\nabla_{a}\left(\tilde{\beta}\tilde{\phi}-\tilde{\alpha}\left(\tilde{\Sigma}-\tfrac{2}{3}\tilde{\Theta}\right)\right).
\end{align}
For Region 1, \eqref{a} becomes 
\begin{align}\label{b}
\alpha\dot{\beta}\phi+\alpha\beta\dot{\phi}+\beta\hat{\beta}\phi+\beta^2\hat{\phi}-\alpha^2\left(\dot{\Sigma}-\tfrac{2}{3}\dot{\Theta}\right)\nonumber\\
-\alpha\beta\left(\hat{\Sigma}-\tfrac{2}{3}\hat{\Theta}\right)-\hat{\alpha}\beta\left({\Sigma}-\tfrac{2}{3}{\Theta}\right)-\dot{\alpha}\alpha\left(\Sigma-\tfrac{2}{3}\Theta\right).
\end{align}
And for Region 2, \eqref{a} becomes
\begin{align}\label{e}
\tilde{\alpha}\mathring{\tilde{\beta}}\tilde{\phi}+\tilde{\alpha}\tilde{\beta}\mathring{\tilde{\phi}}+\tilde{\beta}\bar{\tilde{\beta}}\tilde{\phi}+\tilde{\beta}^2\bar{\tilde{\phi}}-\tilde{\alpha}^2\left(\mathring{\tilde{\Sigma}}-\tfrac{2}{3}\mathring{\tilde{\Theta}}\right)\nonumber\\
-\tilde{\alpha}\tilde{\beta}\left(\bar{\tilde{\Sigma}}-\tfrac{2}{3}\bar{\tilde{\Theta}}\right)-\tilde{\bar{\alpha}}\tilde{\beta}\left(\tilde{\Sigma}-\tfrac{2}{3}\tilde{\Theta}\right)-\mathring{\tilde{\alpha}}\tilde{\alpha}\left(\tilde{\Sigma}-\tfrac{2}{3}\tilde{\Theta}\right).
\end{align}
We now apply  the field equations obtain
\begin{equation}\label{C1}
\mathbb{X}=\tilde{\mathbb{X}},
\end{equation}
with the constraints $\beta=\pm\sqrt{1+\alpha^2}$ and $\tilde\beta=\pm\sqrt{1+\tilde\alpha^2}$. Here
\begin{eqnarray}\label{c}
\mathbb{X}&=&\alpha\dot{\beta}\phi+\alpha\beta\left[-\left(\Sigma-\tfrac{2}{3}\Theta\right)(\mathcal{A}-\tfrac{1}{2}\phi)+Q\right]+\beta\hat{\beta}\phi
\nonumber\\&&+\beta^2\left[-\tfrac{1}{2}\phi^2+\left(\tfrac{1}{3}\Theta+\Sigma\right)\left(\tfrac{2}{3}\Theta-\Sigma\right)-\tfrac{2}{3}(\mu+\Lambda)-\mathcal{E}-\tfrac{1}{2}\Pi\right]\nonumber\\&&-\alpha^2\left[-\mathcal{A}\phi+2\left(\tfrac{1}{3}\Theta-\tfrac{1}{2}\Sigma\right)^2+\tfrac{1}{3}(\mu+3p-2\Lambda)-\mathcal{E}+\tfrac{1}{2}\Pi\right]\nonumber\\&&+\alpha\beta(\tfrac{3}{2}\phi\Sigma+Q)-\hat{\alpha}\beta\left(\Sigma-\tfrac{2}{3}\Theta\right)-\dot{\alpha}\alpha\left(\Sigma-\tfrac{2}{3}\Theta\right),
\end{eqnarray}
and $\tilde{\mathbb{X}}$ is defined likewise. Thus we see that the scalars equations  (\ref{23}), and  (\ref{C1}), completely defines all the conditions that needs to be satisfied for a consistent matching for all epochs. 

\section{Matching Conditions for LRS-II spacetimes: Spacelike matching surface}\label{five}

In this case, we are looking at the scenario, where two different patches of spacetime is matched across a spacelike hypersurface. Although we generally do not have scenarios where this is applied, we just give the equations for completeness. 
Since in this case $n_a=\alpha u_{a}+\beta e_{a}$ and hence $\tilde n_a=\tilde\alpha \tilde u_{a}+\tilde\beta \tilde e_{a}$ is timelike, we have $\beta=\pm\sqrt{-1+\alpha^2}$ and $\tilde\beta=\pm\sqrt{-1+\tilde\alpha^2}$.
The first fundamental form for Region1 is given by
\begin{align}
h_{ab}&= g_{ab}+n_{a}n_{b}\label{31}\\
&=-(1-\alpha^{2})u_{a}u_{b}+(1+\beta^{2})e_{a}e_{b}+\alpha\beta u_{a}e_{b}+\alpha\beta e_{a}u_{b}+N_{ab}.\label{32}
\end{align}
Likewise the first fundamental form in Region 2 is	
\begin{align}
\tilde{h}_{ab}&=\tilde{g}_{ab}+\tilde{n}_{a}\tilde{n}_{b}\label{33}\\
&=-(1-\tilde{\alpha}^{2})\tilde{u}_{a}\tilde{u}_{b}+(1+\tilde{\beta}^{2})\tilde{e}_{a}\tilde{e}_{b}+\tilde{\alpha}\tilde{\beta }\tilde{u}_{a}\tilde{e}_{b}+\tilde{\alpha}\tilde{\beta } \tilde{e}_{a}\tilde{u}_{b}+\tilde{N}_{ab}.\label{34}
\end{align}
Similarly the second fundamental form is given as 
\begin{eqnarray}
\chi_{ab}&=&u_{a}u_{b}\left[-(1-\alpha^{2})^2\dot{\alpha}-\alpha\beta(1-\alpha^{2})\dot{\beta}+\alpha\beta(1-\alpha^{2})\hat{\alpha}\right.\nonumber\\
&&\left.+\alpha^{2}\beta^{2}\hat{\beta}-\beta \mathcal{A}(1-\alpha^{2})+\alpha\beta^{2}\left(\tfrac{1}{3}\Theta+\Sigma\right)\right]+e_{a}e_{b}\left[-\alpha^{2}\beta^{2}\dot{\alpha}\right.\nonumber\\
&&\left.+\alpha\beta(1+\beta^{2})\dot{\beta}-\alpha\beta(1+\beta^{2})\hat{\alpha}+(1+\beta^{2})^{2}\hat{\beta}+\alpha^{2}\beta\mathcal{A}\right.\nonumber\\
&&\left.+\alpha(1+\beta^{2})(\tfrac{1}{3}\Theta+\Sigma)\right]+u_{(a}e_{b)}\left[2\alpha\beta(1-\alpha^{2})\dot{\alpha}+2\alpha\beta(1+\beta^2)\hat{\beta}\right.\nonumber\\
&&\left.+((1-\alpha^{2})(1+\beta^{2})-\alpha^2\beta^2)\hat{\alpha}+(\alpha^2\beta^2-(1-\alpha^{2})(1+\beta^{2}))\dot{\beta}\right.\nonumber\\
&&\left.+(\tfrac{1}{3}\Theta+\Sigma)(\beta(1+\beta^2)+\alpha^2\beta)+\alpha\mathcal{A}(\beta^2-(1-\alpha^2))\right]\nonumber\\
&&+N_{ab}\left[\alpha\left(\tfrac{1}{3}\Theta-\tfrac{1}{2}\Sigma\right)+\tfrac{1}{2}\beta\phi\right]
\end{eqnarray}
and $\tilde{\chi}_{ab}$ is given likewise. Thus the matching condition becomes
\begin{align}
\beta\phi-\alpha\left(\Sigma-\tfrac{2}{3}\Theta\right)&=\tilde{\beta}\tilde{\phi}-\tilde{\alpha}\left(\tilde{\Sigma}-\tfrac{2}{3}\tilde{\Theta}\right)\label{23a},
\end{align}
with the constraints  $\beta=\pm\sqrt{-1+\alpha^2}$ and $\tilde\beta=\pm\sqrt{-1+\tilde\alpha^2}$.  It is interesting to note that the consistency of the constraint remains same as equation (\ref{C1}), with the new relation between $\alpha$ and $\beta$. Thus equations (\ref{23a}) and (\ref{C1}) with $\beta=\pm\sqrt{-1+\alpha^2}$ and $\tilde\beta=\pm\sqrt{-1+\tilde\alpha^2}$ completely determines the matching conditions and their consistancy.

\section{A well known example: Matching a radiating star with a Vaidya exterior}

To illustrate the advantage of the semi-tetrad 1+1+2 matching equations, let us revisit the well known scenario of matching a radiating and collapsing spherically symmetric star, having a comoving boundary, with a Vaidya exterior. 
Let the metric in the interior spacetime (in the comoving coordinates ($t,r,\theta,\phi$)) be given as
\be
ds_1^2=-A(t,r)dt^2+B(t,r)dr^2+R^2(t,r)(d\theta^2+\sin^2(\theta)d\phi^2)\;.
\ee
Here $R(t,r)$ is the area radius of the collapsing 2-sheets and hence we have
\be
K=\frac{1}{R^2(t,r)}.
\ee
We consider the boundary of the star to be the comoving shell labelled by $r=r_b$ in the interior spacetime. Let this star be matched to the Vaidya exterior, with the metric 
\be
ds_2^2= -\left(1-\frac{2m(v)}{r_v}\right)dv^2-2dvdr_v+r_v^2(d\theta^2+\sin^2(\theta)d\phi^2)\;.
\ee
Here $v$ is the exploding null coordinate and $r_v$ is the Vaidya radius. The boundary as seen from the exterior spacetime is given as $r_{vb}=r_{vb}(v)$. The unit vectors in the $[u,e]$ plane for the Vaidya spacetime is given as
\be
\tilde u^a=\left(1-\frac{2m(v)}{r_v}\right)^{-1/2}\left(\frac{\partial}{\partial v}\right)^a,
\ee
and
\be
\tilde e^a=-\left(1-\frac{2m(v)}{r_v}\right)^{-1/2}\left(\frac{\partial}{\partial v}\right)^a+\left(1-\frac{2m(v)}{r_v}\right)^{1/2}\left(\frac{\partial}{\partial r_v}\right)^a.
\ee
It can be easily checked that for the above unit vectors, we have the following
\begin{eqnarray}
\tilde\Sigma-\frac23\tilde\Theta&=&0,\\
\tilde\mu=3\tilde p=\tilde Q&=&\frac32\tilde\Pi,\\
\mathring{v}=-\bar{v}.\label{vdot}
\end{eqnarray}
Now, since the boundary is comoving in the interior spacetime, the normal to the boundary is purely along the $e$-direction. Therefore we have $\alpha=0$ and $\beta=1$. In the exterior spacetime the normal lying in the  $[u,e]$ plane is null, and therefore we must have $\tilde{n_a}\tilde{n^a}=0$. This can then be normalised to $\tilde\alpha=\tilde\beta=1/\sqrt{2}$. Plugging these in (\ref{23}) and (\ref{match2}) we get for the boundary
\be\label{match3}
\frac{\phi}{(\Sigma-\tfrac23\Theta)}=1\;,
\ee
that must be satisfied as we approach the boundary in the interior spacetime. It can again be easily calculated, that the Misner Sharp mass of the Vaidya exterior is given as 
\be
\tilde{\mathcal{M}}=m(v).
\ee
Now, here comes an important physical observation:  at the boundary, the rate of change of Misner-Sharp mass, along the fluid flow lines must be same on both sides. In other words, the total mass lost in the interior spacetime must be radiated away along the outgoing null geodesics of Vaidya exterior. Since at the exterior spacetime the boundary is comoving, therefore this will give, 
\be
u^a\nabla_a{\mathcal{M}}=\left(\frac{1}{\sqrt2}\tilde{u^a}+\frac{1}{\sqrt2}\tilde{e^a}\right)\nabla_a m(v).
\ee
Using (\ref{vdot}) and the field equations, the above expression immediately simplifies to
\be\label{match4}
\frac{\phi}{(\Sigma-\tfrac23\Theta)}=\frac{Q}{p}\;,
\ee
which using (\ref{match3}) gives us the important result as found by Santos: At the comoving boundary of a collapsing radiating star, we must have
\be
p=Q,
\ee
that is, the isotropic pressure must be equal to the radial heat flux at the comoving boundary, if we need to match the star with an exploding Vaidya exterior. From, the above equation it is clear that if there is no heat flux in the interior spacetime, the pressure at the comoving boundary {\em must} be zero, if the interior is matched to Vaidya. On the contrary, zero pressure for Vaidya spacetime necessarily implies $m(v)_{,v}=0$, or constant Misner-Sharp mass., which in turn implies that the exterior spacetime must be Schwarzschild in that case. When the pressures in the radiating star is non-isotropic, that is the radial pressure $p_r$ is not equal to the tangential pressure $p_{\theta}$, the above condition can be written as
\be
p_r=Q-\frac23\Delta,
\ee
where $\Delta=p_r-p_{\theta}$, is the anisotropy parameter. 

 \section{Discussion}\label{five}
 
 In this paper we have transparently showed the following: The Israel-Darmois matching conditions across a general hypersurface for LRS-II spacetimes simplifies to three scalar equations in the 1+1+2 semitetrad formalism. These scalar equations give us the relations between the geometrical and thermodynamic quantities of both sides of the hypersurface and also relates these to the dynamics of the hypersurface itself. We derived the equations for cases when the normal to these hypersurface is spacelike and also timelike. Note that a careful matching of these two cases will then generate the matching condition across a null hypersurface which is more complicated, which we have not considered in this paper. \\
 
Writing the matching conditions in terms of the scalar equation in the semitetrad formalism has a number of advantages. First and foremost, this gives a direct relation between the components of energy momentum tensors of the spacetime patches on either side of the matching hypersurface. This is very useful while matching spherically symmetric spacetimes across a stellar surface (for example), where the spacetime just outside a stellar surface is not vacuum. Also these equations can be used to model multi-regions in spherically symmetric stellar structures or even gravastars, where the energy momentum tensor of each region is different from the other. \\

Our result gives another important realisation. One patch of given spacetime with LRS-II symmetry can in principle be matched with a large number of different LRS-II spacetimes, by carefully choosing the dynamics of the matching hypersurface.  Thus, if we consider the surface of a collapsing spherical star to be non-comoving, we can match the collapsing star with a number of different exterior spacetimes and that will lead the collapse to have different end states.

\section {Acknowledgement \label{acknowledge}}
PK and RG are indebted to the National Research Foundation and the University of KwaZulu-Natal for financial support.
SDM acknowledges that this work is based upon research supported by the South African Research Chair Initiative of the Department of Science and Technology.

 \end{document}